\documentclass{article}
\usepackage{amsthm}
\usepackage{amssymb}
\usepackage{amsmath}
\usepackage{epsfig}
\usepackage{colonequals}
\usepackage{enumerate}
\usepackage{hyperref}
\usepackage[dvips]{color}

\newtheorem{theorem}{Theorem}[section]
\newtheorem{corollary}[theorem]{Corollary}
\newtheorem{lemma}[theorem]{Lemma}
\newtheorem{observation}[theorem]{Observation}

\newcount\claimno
\def\newclaim#1#2{
   \global\advance\claimno by 1\relax
   \bigskip\noindent\rlap{\rm(\the\claimno)}\ignorespaces
   \global\expandafter\edef\csname CLAIMLABEL#1\endcsname{\the\claimno}\relax
   \hangindent=33pt\hskip30pt{\sl#2}\bigskip}
\def\mylabel#1{{\label{#1}}}
\def\junk#1{}

\begin{document}
\title{Three-coloring triangle-free graphs on surfaces VII.  A linear-time algorithm\thanks{September 2020}}
\author{%
     Zden\v{e}k Dvo\v{r}\'ak\thanks{Computer Science Institute (CSI) of Charles University,
           Malostransk{\'e} n{\'a}m{\v e}st{\'\i} 25, 118 00 Prague, 
           Czech Republic. E-mail: {\tt rakdver@iuuk.mff.cuni.cz}.
	   Supported by project GA14-19503S (Graph coloring and structure) of Czech Science Foundation.}
 \and
     Daniel Kr{\'a}l'\thanks{Faculty of Informatics,
            Masaryk University, Botanick\'a 68A, 602 00 Brno, Czech Republic, and
            Mathematics Institute, DIMAP and Department of Computer Science, University
            of Warwick, Coventry CV4 7AL, UK. E-mail: {\tt dkral@fi.muni.cz}. The
            previous affiliation: Computer Science Institute (CSI) of Charles University.}
 \and
        Robin Thomas\thanks{School of Mathematics, 
        Georgia Institute of Technology, Atlanta, GA 30332. 
        E-mail: {\tt thomas@math.gatech.edu}.
        Partially supported by NSF Grants No.~DMS-0739366 and DMS-1202640.}
}
\date{}
\maketitle
\begin{abstract}
We give a linear-time algorithm to decide $3$-colorability of a triangle-free graph embedded in a fixed surface,
and a quadratic-time algorithm to output a $3$-coloring in the affirmative case.
The algorithms also allow to prescribe the coloring of a bounded number of vertices.
\end{abstract}

\section{Introduction}

This paper is the last part of a series aimed at studying the $3$-colorability
of graphs on a fixed surface that are either triangle-free, or have their
triangles restricted in some way (throughout the paper, all colorings are proper, i.e., adjacent vertices
have different colors).  The main result of this paper is a linear-time algorithm to decide $3$-colorability of a triangle-free graph
embedded in a fixed surface.
Embeddability in a surface is not a sufficient restriction by itself, as
$3$-colorability of planar graphs is NP-complete~\cite{garey1979computers}.  Restricting the triangles
is natural in the light of the well-known theorem of Gr\"otzsch~\cite{grotzsch1959} stating that every planar triangle-free graph is $3$-colorable.

A graph $G$ is \emph{$4$-critical} if every proper subgraph of $G$ is $3$-colorable but $G$ itself is not.
Clearly, a graph is $3$-colorable if and only if it has no $4$-critical subgraph.
As was shown by Thomassen~\cite{thomassen-surf} and later (with better bounds) by us~\cite{trfree3},
for every surface $\Sigma$, there are only finitely many $4$-critical graphs of girth at least $5$ that can be embedded in $\Sigma$.
Hence, to decide whether a graph of girth at least $5$ embeddable in $\Sigma$ is $3$-colorable, it suffices to test the presence
of these finitely many subgraphs, which can be done in linear time using the algorithm of Eppstein~\cite{bib-eppstein99}.

The situation is more complicated for triangle-free graphs.  
The Mycielski graph of an odd cycle embeds in any surface other than the sphere.  Furthermore,
Youngs~\cite{Youngs} gave more general infinite families of
$4$-critical triangle-free graphs embeddable in any non-orientable surface.
However, in~\cite{trfree4}
we showed that most of the faces of 4-critical triangle-free graphs drawn in a fixed surface are of length $4$.  In order to state the result precisely,
let us first give some definitions.

A \emph{surface} is a compact connected $2$-manifold with (possibly null) boundary.
In a graph drawn in a surface with a boundary, we require that the drawing of each edge is either completely contained in the boundary,
or disjoint from the boundary except possibly for its endpoints.
Each component of the boundary of a surface is homeomorphic to a circle, and we call it a \emph{cuff}.  For non-negative integers $a$, $b$ and $c$,
let $\Sigma(a,b,c)$ denote the surface obtained from the sphere by adding $a$ handles, $b$ crosscaps and
removing the interiors of $c$ pairwise disjoint closed discs.  The classification theorem of surfaces shows that
every surface is homeomorphic to $\Sigma(a,b,c)$ for some choice of $a$, $b$ and $c$.
The \emph{Euler genus} $g(\Sigma)$ of a surface $\Sigma$ homeomorphic to $\Sigma(a,b,c)$ is defined as $2a+b$.
Consider a graph $G$ embedded in the surface $\Sigma$; when useful, we identify $G$ with the topological
space consisting of the points corresponding to the vertices of $G$ and the simple curves corresponding
to the edges of $G$.  A \emph{face} $f$ of $G$ is a maximal connected subset of $\Sigma-G$.
By the \emph{length} $|f|$ of $f$, we mean the sum of the lengths of the boundary walks of $f$ (in particular, if an edge
appears twice in the boundary walks, it contributes $2$ to $|f|$).
A face $f$ is \emph{$2$-cell} if it is homeomorphic to an open disk, and it is \emph{closed $2$-cell} if additionally
its boundary forms a cycle in $G$.

Finally, we are ready to state the result from~\cite{trfree4}.

\begin{theorem}[{\cite[Theorem~1.3]{trfree4}}]\mylabel{thm:ctvrty}
There exists a constant $\kappa$ with the following property.  Let $G$ be a graph embedded in
a surface of Euler genus $g$.  Let $t$ be the number of triangles in $G$ and let $c$ be the number of
$4$-cycles in $G$ that do not bound a $2$-cell face.  If $G$ is $4$-critical, then
$$\sum_{\text{$f$ face of $G$}} (|f|-4)\le \kappa(g+t+c-1).$$
\end{theorem}

Furthermore, in the previous paper of the series~\cite{trfree6}, we designed a $3$-coloring algorithm for
graphs with almost all faces of length $4$.  
A graph $H$ is a \emph{quadrangulation} of a surface $\Sigma$ if all faces of $H$ are closed $2$-cell and have length $4$
(in particular, the boundary of $\Sigma$ is formed by a set of pairwise vertex-disjoint cycles in $H$, called the \emph{boundary cycles} of $H$).
A vertex of $G$ contained in the boundary of $\Sigma$ is called a \emph{boundary vertex}.

\begin{theorem}[{\cite[Theorem~1.3]{trfree6}}]\label{thm-quadalg}
For every surface $\Sigma$ and integer $k$, there exists a linear-time algorithm with input
\begin{itemize}
\item $G$: a quadrangulation of $\Sigma$ with at most $k$ boundary vertices, and
\item $\psi$: a function from boundary vertices of $G$ to $\{1,2,3\}$,
\end{itemize}
which correctly decides whether there exists a $3$-coloring $\varphi$ of $G$ such that $\varphi(v)=\psi(v)$ for every boundary vertex $v$ of $G$.
In the affirmative case, the algorithm also outputs such a coloring $\varphi$.
\end{theorem}

By combining Theorems~\ref{thm:ctvrty} and \ref{thm-quadalg}, we obtain a straightforward algorithm to test $3$-colorability
of a triangle-free graph embedded in a fixed surface, at least under the assumption that all $4$-cycles in $G$ bound $2$-cell
faces---enumerate all subgraphs $H$ of $G$ such that 
\begin{equation}\label{eq-shf}
\sum_{\text{$h$ face of $H$}} (|h|-4)\le \kappa g
\end{equation}
and test whether they are $3$-colorable.  It can be shown that there are at most $|V(G)|^{5\kappa g}$ subgraphs of $G$ satisfying (\ref{eq-shf}),
and thus for any fixed surface, we obtain a polynomial-time algorithm.  However, the exponent of the polynomial bounding the complexity of
this algorithm depends on the surface.  In this paper, we use a more involved argument to design a linear-time algorithm deciding
$3$-colorability.  Furthermore, similarly to Theorem~\ref{thm-quadalg}, we can allow a bounded number of precolored vertices in
the considered graph.  The following is our first main result.

\begin{theorem}\label{thm-mainalg}
For every surface $\Sigma$ and integer $k$, there exists a linear-time algorithm with input
\begin{itemize}
\item $G$: a triangle-free graph embedded in $\Sigma$ with at most $k$ boundary vertices, and
\item $\psi$: a function from boundary vertices of $G$ to $\{1,2,3\}$,
\end{itemize}
which correctly decides whether there exists a $3$-coloring $\varphi$ of $G$ such that $\varphi(v)=\psi(v)$ for every boundary vertex $v$ of $G$.
\end{theorem}
The algorithm of Theorem~\ref{thm-mainalg} assumes the embedding of $G$ in the surface $\Sigma$ is given as a part of the input.
However, since the surface $\Sigma$ is fixed, such an embedding of an abstract graph can be found (or shown not to
exist) in linear time using an algorithm of Mohar~\cite{mohar1999linear}.  In order to ensure that the precolored vertices are incident with the
boundary, we can then drill holes next to them, if needed.  We also need to specify how the embedding of $G$ is given.
We use a variant of the polygonal representation---we choose a ``cut graph'' $H$, where $H$ has exactly one face, this face is $2$-cell,
every cuff of $\Sigma$ is traced by a cycle in $H$, and every edge of $H$ is either equal to an edge of $G$, or its drawing intersects $G$ only in vertices.
Then, we cut $\Sigma$ and $G$ along the edges of $H$, and we represent the resulting graph drawn in a disk.  More details can be found in~\cite{trfree6},
where we developed all the subroutines for which the details of the representation are relevant.

Note that unlike Theorem~\ref{thm-quadalg}, Theorem~\ref{thm-mainalg} does not return a $3$-coloring
if one exists; indeed, the algorithm only decides whether there exists a critical subgraph.
Let us remark that this a quite common situation; even in the case of planar graphs, a linear-time algorithm to actually find a $3$-coloring
guaranteed by Gr\"otzsch's theorem was only designed recently~\cite{DvoKawTho}.
In Section~\ref{sec-repalg}, we show that using Theorem~\ref{thm-mainalg}, we can in quadratic time reduce the problem of finding a $3$-coloring
to the case that the graph has girth at least five.  This case can be dealt with by algorithmizing the ideas of~\cite{trfree3},
or more conveniently using an argument of Postle~\cite{postledistr-focs}.
Hence, we obtain the following algorithm.

\begin{theorem}\label{thm-mainalgrep}
For every surface $\Sigma$ and integer $k$, there exists a quadratic-time algorithm with input
\begin{itemize}
\item $G$: a triangle-free graph embedded in $\Sigma$ with at most $k$ boundary vertices, and
\item $\psi$: a function from boundary vertices of $G$ to $\{1,2,3\}$,
\end{itemize}
which correctly decides whether there exists a $3$-coloring $\varphi$ of $G$ such that $\varphi(v)=\psi(v)$ for every boundary vertex $v$ of $G$,
and outputs such a coloring in the affirmative case.
\end{theorem}

Let us remark that we believe that there exists a linear-time algorithm to output a $3$-coloring using ideas similar to those
of Theorem~\ref{thm-mainalg}; however, there are significant technical challenges in designing it and we leave this as an open
problem.  Let us also remark that in Theorems~\ref{thm-mainalg} and \ref{thm-mainalgrep}, it would suffice to only forbid the existence of contractible
triangles, as we can deal with the non-contractible ones by cutting the surface along them (see Theorem~\ref{thm-mainalg-spec4} for
a similar idea used to eliminate non-contractible $4$-cycles).

In the following section, we recall the results and definitions from the previous papers of the series we are going to need.
In Section~\ref{sec-freedom}, we define a key notion of a free set of faces and apply it to the special case of graphs
embedded in the disk.
In Section~\ref{sec-linalg}, we give the linear-time decision algorithm.  Finally, in Section~\ref{sec-repalg}, we give the
algorithm to output a $3$-coloring if one exists.

\section{Definitions and previous results}

We need a stronger form of Theorem~\ref{thm:ctvrty} which deals with graphs with precolored cycles.  First, let us give several definitions.
Suppose that a graph $G$ is embedded in a surface $\Sigma$ so that every cuff of $\Sigma$ traces a cycle in $G$,
let $H$ be a subgraph of $G$, and let $h$ be a face of $H$.  We would like to view the part $G'$ of $G$ drawn in the closure
of $h$ as drawn in the surface $\Sigma_h$ whose interior is homeomorphic to $h$.  There is a minor technical issue to overcome;
consider e.g. the case that $\Sigma$ is the torus and $h$ is an open cylinder bounded by two non-contractible homotopic cycles
intersecting in a path $Q$.  Then, to view $G'$ as a graph drawn in the cylinder, we need to cut the surface along the two cycles,
naturally splitting $Q$ into two paths in the process.  More precisely, we let $\theta_h:\Sigma_h\to \Sigma$ be a continuous function whose restriction to the interior of $\Sigma_h$
is a homeomorphism to $h$, and we define $G_h=\theta_h^{-1}(G)$.

Let $s:\mathbf{Z}^+\to\mathbf{R}$ be defined by
$$s(n)=\begin{cases}
0&\text{if } n\le 4\\
4/4113&\text{if } n=5\\
72/4113&\text{if } n=6\\
540/4113&\text{if } n=7\\
2184/4113&\text{if } n=8\\
n-8&\text{if } n\ge 9.
\end{cases}$$
To each $2$-cell face $f$ of $G$, we assign a weight $w_0(f)=s(|f|)$.
If $f$ is not $2$-cell, then let $w_0(f)=|f|$.
For a surface $\Pi$ of Euler genus $g$ with $c$ cuffs, let $s(\Pi)=6c-6$ if $g=0$ and $c\le 2$, and $s(\Pi)=120g+48c-120$ otherwise.
For a real number $\eta$ and a face $f$ of $G$, let $w_\eta(f)=w_0(f)+\eta s(\Sigma_f)$.
Let $$w_\eta(G)=\sum_{\text{$f$ face of $G$}} w_\eta(f).$$

\begin{theorem}[{\cite[Lemma~5.2]{trfree6}}]\label{thm-almstruct}
There exists a constant $\eta>0$ such that the following holds.
Let $G$ be a triangle-free graph embedded in a surface $\Sigma$ without non-contractible $4$-cycles, so that every cuff of $\Sigma$ traces a cycle in $G$, and
let $B$ be the union of boundary cycles of $G$.  There exists a subgraph $H$ of $G$ such that $B\subseteq H$,
$w_\eta(H)\le w_\eta(B)$ and for every face $h$ of $H$, every $3$-coloring of the boundary of $h$ extends to a $3$-coloring of $G_h$.
\end{theorem}

Throughout the rest of the paper, let $\eta$ denote the constant of Theorem~\ref{thm-almstruct}.
We need a stronger variant for the disk, see Corollary~5.3 in~\cite{trfree4}.

\begin{theorem}\label{thm-diskcase}
Let $G$ be a triangle-free graph embedded in the disk with boundary cycle $B$.
Then either every $3$-coloring of $B$ extends to a $3$-coloring of $G$, or
there exists a connected subgraph $H\subseteq G$ such that $B\subsetneq H$ and $w_0(H)\le s(|B|-2)$.
\end{theorem}
\begin{proof}
If some $3$-coloring $\varphi$ of $B$ does not extend to a $3$-coloring of $G$, then let $H\supsetneq B$
be a minimal subgraph of $G$ such that $\varphi$ does not extend to a $3$-coloring of $H$.
By Gr\"otzsch's theorem and the minimality of $H$, we conclude that $H$ is connected.
Furthermore, it is easy to see that $H$ is $B$-critical (in the sense defined in Section 2 of~\cite{trfree4}),
and by \cite[Corollary~5.3]{trfree4}, we have $w_0(H)\le s(|B|-2)$.
\end{proof}

Consider a graph $G$ embedded in a surface $\Sigma$.
A cycle $K$ in $G$ is \emph{contractible} if there exists a closed disk
$\Delta\subseteq \Sigma$ with boundary equal to $K$.  For a cuff $C$, let
$\Sigma+\widehat{C}$ denote the surface obtained from $\Sigma$ by adding an open
disk disjoint from $\Sigma$ and with boundary equal to $C$; we say that
$\Sigma+\widehat{C}$ is obtained from $\Sigma$ by \emph{patching a cuff}.  A cycle $K$
\emph{surrounds a cuff $C$} if
$K$ is not contractible in $\Sigma$, but it is contractible in $\Sigma+\widehat{C}$.

We use a data structure we designed in~\cite{trfree6}, see Lemma~4.6 (in this paper, we do not
use the operation of contracting the edges of a star forest, and thus we omit the relevant parts of the statement).

\begin{lemma}\label{lemma-dataemb}
For any integer $d\ge 0$ and every surface $\Sigma$, there exists a data structure as follows.
The data structure represents a graph $G$ with a $2$-cell embedding in $\Sigma$ and supports the
following operations in amortized constant time (depending only on $d$ and $\Sigma$):
\begin{enumerate}[(a)]
\item Removal of an edge or an isolated vertex.
\item For any vertex $v\in V(G)$, deciding whether there exists a closed walk $W$ of length at most $d$ with $v\in V(W)$ such that
$W$ is not null-homotopic even after patching any one cuff of $\Sigma$ with a disk, and finding such a walk if that is the case.
\item For any vertex $v\in V(G)$ and any set $D$ of cuffs of $\Sigma$, letting $\Sigma'$ be a surface obtained from $\Sigma$ by
patching all the cuffs in $D$ and letting $\Lambda\subseteq\Sigma'$ be an open disk containing all the patches,
deciding whether there exists a closed walk $W$ in $G$ of length at most $d$ such that $W$ contains $v$ and is homotopically
equivalent (in $\Sigma$) to the boundary of $\Lambda$, and finding such a walk if that is the case.
\end{enumerate}
The data structure can be initialized in $O(|V(G)|)$ time.
\end{lemma}

In~\cite{trfree6}, we designed several useful algorithms.  One of them can be used to eliminate contractible $(\le\!4)$-cycles.

\begin{lemma}[\cite{trfree6}, Lemma~4.9]\label{lemma-elim24}
For any surface $\Sigma$, there exists a linear-time algorithm that,
given a graph $G$ with a $2$-cell embedding in $\Sigma$ such that every cuff of $\Sigma$ traces a cycle in $G$,
returns a subgraph $H$ of $G$ such that
\begin{itemize}
\item $H$ (with its drawing inherited from $G$) is $2$-cell embedded in $\Sigma$ and all boundary cycles of $G$ belong to $H$,
\item all contractible cycles in $H$ of length at most $4$ bound $2$-cell faces, and
\item all vertices and edges of $G$ that do not belong to $H$ are drawn in $2$-cell $(\le\!4)$-faces of $H$.
\end{itemize}
\end{lemma}

We will need this result in combination with the algorithm of Dvo\v{r}\'ak, Kawarabayashi and Thomas~\cite{DvoKawTho}.

\begin{theorem}\label{thm-ext45}
There exists a linear-time algorithm as follows.  Let $G$ be a plane triangle-free graph with the outer face bounded by
a cycle $C$ of length at most $5$.  Given a $3$-coloring $\psi$ of $C$, the algorithm returns a $3$-coloring of $G$ that
extends $\psi$.
\end{theorem}

In particular, in the situation of Lemma~\ref{lemma-elim24}, any $3$-coloring of $H$ can be extended to a $3$-coloring of $G$
in time $O(|V(G)\setminus V(H)|)$.

Consider a graph $G$ embedded in a surface $\Sigma$.
A subgraph $H$ of $G$ is \emph{non-essential}
if there exists $\Lambda\subset \Sigma$ containing $H$, where $\Lambda$ is either an open disk, or an open disk with a hole whose boundary is equal to a cuff of $\Sigma$.
A subgraph $H$ of $G$ is \emph{essential} if it is not non-essential.
We say that a surface $\Sigma'$ is \emph{at most as complex as $\Sigma$} if $\Sigma'$ has smaller genus than $\Sigma$, or
$\Sigma'$ has the same genus and fewer cuffs than $\Sigma$, or $\Sigma'$ is homeomorphic to $\Sigma$.
For a graph $G$ embedded in $\Sigma$ so that every cuff of $\Sigma$ traces a cycle in $G$, let $b(G)$ denote the multiset of the lengths of the boundary cycles of $G$.
For two multisets $S$, $T$ of integers such that $|S|=|T|=m$, we say that \emph{$S$ dominates $T$} if
there exists an ordering $s_1,\ldots, s_m$ of the elements of $S$ and an ordering $t_1,\ldots,t_m$ of the elements of $T$
such that $s_i\ge t_i$ for $i=1,\ldots, m$. 

The following algorithm is useful when dealing with essential subgraphs; let us recall that the notations
$\Sigma_h$ and $G_h$ were defined at the beginning of this section.

\begin{lemma}[\cite{trfree6}, Lemma~4.7]\label{lemma-qstruc}
For any function $\nu(\Pi,n)$, any surface $\Sigma$ and any integer $k\ge 0$, there exists a constant $\sigma$ and a linear-time algorithm as follows.
Let $G$ be a graph $2$-cell embedded in $\Sigma$ with boundary cycles $B_1$, \ldots, $B_c$ of total length at most $k$.
The algorithm returns a subgraph $H$ of $G$ with at most $\sigma$ vertices such that $B_1\cup \ldots\cup B_c\subseteq H$
and for each face $h$ of $H$, $G_h$ (in its embedding in $\Sigma_h$) does not contain any
connected essential subgraph with fewer than $\nu(\Sigma_h,k_h)$ edges, where $k_h$ is the sum of the lengths
of the boundary cycles of $G_h$.
Furthermore, $\Sigma_h$ is at most as complex as $\Sigma$, and if $\Sigma_h$ is homeomorphic to $\Sigma$, then
$b(G)$ dominates $b(G_h)$.
\end{lemma}

We also need a similar algorithm to deal with the cylinder case, to obtain a maximal ``laminar'' set of
short non-contractible cycles.

\begin{lemma}[\cite{trfree6}, Lemma~4.3]\label{lemma-splitcyl}
Let $d$ be a positive integer.  There exists a linear-time algorithm that,
given a graph $G$ that is $2$-cell embedded in the cylinder $\Sigma$ with boundary cycles $B_1$ and $B_2$
of length at most $d$, returns a sequence $C_0$, $C_1$, \ldots, $C_m$ of non-contractible cycles of $G$ of length at most $d$
such that
\begin{itemize}
\item $C_0=B_1$ and $C_m=B_2$,
\item for $0\le i< m$, the cycle $C_i$ is contained in the part of $\Sigma$ between $B_1$ and $C_{i+1}$, and
\item either $C_i$ intersects $C_{i+1}$, or the subcylinder of $\Sigma$ between $C_i$ and $C_{i+1}$ contains no
non-contractible cycle of length at most $d$ distinct from $C_i$ and $C_{i+1}$.
\end{itemize}
\end{lemma}

\section{Freedom}\label{sec-freedom}

Let $G$ be a triangle-free graph with a $2$-cell embedding in a surface $\Sigma$.  Let $S$ be a set of faces of $G$ and let $W$ be a contractible
closed walk in $G$ forming the boundary of an open disk $\Lambda\subset\Sigma$.  We say that $W$ \emph{binds} $S$ (with respect to $\Lambda$) if
$\Lambda\not\in S$ and
$$\sum_{f\in S, f\subseteq \Lambda} w_0(f)\ge s(|W|).$$  We say that $S$ is
\emph{$k$-free} if no closed walk of length at most $k$ binds $S$.  The key observation is that the presence of a large free set
ensures the possibility to extend a precoloring.

\begin{lemma}\label{lemma-disk}
Let $G$ be a triangle-free graph with a $2$-cell embedding in the disk $\Delta$, with boundary cycle $B$.  If
$G$ contains a $(|B|-2)$-free set $S$ of faces such that $\sum_{f\in S} w_0(f)>s(|B|-2)$,
then every $3$-coloring of $B$ extends to a $3$-coloring~of~$G$.
\end{lemma}
\begin{proof}
Suppose for a contradiction that some $3$-coloring of $B$ does not extend to a $3$-coloring of $G$.
By Theorem~\ref{thm-diskcase}, there exists a connected subgraph $H\subseteq G$ with $B\subsetneq H$ such that $w_0(H)\le s(|B|-2)$.
Note that $|h|\le |B|-2$ for every face $h$ of $H$.  Since $S$ is $(|B|-2)$-free,
we have $$\sum_{f\in S, f\subseteq h} w_0(f)\le w_0(h),$$ where we only write the non-strict inequality since $h$ may belong to $S$.
Therefore, $$\sum_{f\in S} w_0(f)\le \sum_{\text{$h$ face of $H$}} w_0(h)=w_0(H)\le s(|B|-2),$$
which is a contradiction.
\end{proof}

We will need the following consequence.  If $G$ is a graph embedded in $\Sigma$ and $\Lambda$ is an open disk whose boundary is
contained in $G$, then let $G-\Lambda$ be the graph obtained from $G$ by removing all vertices and edges drawn in $\Lambda$.

\begin{corollary}\label{cor-surfdisk}
Let $G$ be a triangle-free graph with a $2$-cell embedding in a surface $\Sigma$, let $W$ be a closed walk in $G$ bounding
an open disk $\Lambda\subset \Sigma$, and let $S$ be a set of faces such that
$\sum_{f\in S,f\subseteq\Lambda} w_0(f)>s(|W|-2)$.
If $S$ is $(|W|-2)$-free, then every $3$-coloring of $G-\Lambda$ extends to a $3$-coloring of $G$.
\end{corollary}

Let us design an algorithm enabling us to perform this reduction efficiently.
We need an observation enabling us to simplify closed walks of a given homotopy.
By \emph{drilling a hole} in a face $f$ of an embedded graph, we mean deleting the interior of an arbitrary closed disk contained
in $f$ (so we obtain a new cuff disjoint from the boundary of $f$).

\begin{lemma}\label{lemma-simpc}
Let $G$ be a graph with a $2$-cell embedding in a surface $\Sigma$ other than the sphere, let $S$ be a non-empty set of faces of $G$
and let $\Sigma'$ be the surface obtained from $\Sigma$ by drilling holes in the faces of $S$.
Let $\Lambda\subseteq \Sigma$ be an open disk such that $\Sigma\setminus\Sigma'\subseteq \Lambda$.
Let $W$ be a closed walk in $G$ homotopically equivalent in $\Sigma'$ to the boundary of $\Lambda$,
and let $H_W$ be the subgraph of $G$ consisting of the vertices and the edges of $W$.
Consider the drawing of $H_W$ in $\Sigma$ inherited from $G$, and let $T$ be the set of faces of $H_W$
that intersect $\Sigma\setminus\Sigma'$.

If $G$ (in $\Sigma$) contains no connected essential subgraph with at most $|W|$ edges, then each face of $T$ is $2$-cell and $\sum_{f\in T} |f|\le|W|$.
\end{lemma}
\begin{proof}
Since $G$ contains no connected essential subgraph with at most $|W|$ edges, we conclude that there exists $\Lambda_H\subseteq \Sigma$ containing
$H_W$ such that $\Lambda_H$ is either an open disk, or an open disk with a hole whose boundary is equal to a cuff $C$ of $\Sigma$.
In the latter case, let $\Lambda'_H$ be the open disk obtained from $\Lambda_H$ by patching the hole corresponding to $C$.
In the former case, let $\Lambda'_H=\Lambda_H$.

Consider $H_W$ as embedded in $\Lambda'_H$.
Let $f_0$ be the face of $H_W$ containing the boundary of $\Lambda'_H$,
and if $\Lambda'_H\neq\Lambda_H$, then let $f_1$ be the face of $H_W$ containing $\Lambda'_H\setminus\Lambda_H$.
Since $H_W$ is connected, all faces of $H_W$ except for $f_0$ are $2$-cell.  Since $W$ is homotopically
equivalent to the boundary of $\Lambda$, we have $f_0,f_1\not\in T$.  Consequently, all faces of $T$ are $2$-cell.
Furthermore, if two faces of $T$ share an edge $e$, then $e$ appears in $W$ at least twice (since the boundary of $\Lambda$,
and thus also $W$, has the same winding number $\pm 1$ around both of the faces).  Therefore, $\sum_{f\in T} |f|\le|W|$.
\end{proof}

We can now design the following subroutine.

\begin{lemma}\label{lemma-testfree1}
For any integer $k\ge 0$ and a surface $\Sigma$ other than the sphere, there exists a linear-time algorithm as follows.
Let $G$ be a triangle-free graph with a $2$-cell embedding in $\Sigma$ such that every connected essential
subgraph of $G$ has more than $k$ edges, and let $f$ be a face of $G$.
The algorithm decides whether $\{f\}$ is $k$-free, and if not, returns a closed walk $W$ of length at most $k$
that binds $\{f\}$ with respect to an open disk $\Lambda\subseteq\Sigma$ such that among all such walks,
$|W|$ is minimal, and additionally $\{\Lambda\}$ is $k$-free in $G-\Lambda$.
\end{lemma}
\begin{proof}
Let $\Sigma'$ be the surface obtained from $\Sigma$ by drilling a hole inside $f$, and let
$\Lambda'$ be the open disk removed from $\Sigma$ in order to create $\Sigma'$.
Build the data structure of Lemma~\ref{lemma-dataemb} for $G$ with $d=k$,
and find (in linear time) a shortest walk $W_0$ of length at most $k$ such that $W_0$ is homotopically equivalent in $\Sigma'$ to the boundary
of $\Lambda'$.  If no such walk exists, then $\{f\}$ is $k$-free.

Otherwise, let $t=|W_0|$.  Let $H_0$ be the subgraph of $G$ consisting of the vertices and edges of $G$ contained in $W_0$,
and let $\Lambda_0$ be the face of $H_0$ containing $f$.  By Lemma~\ref{lemma-simpc}, $\Lambda_0$ is an open disk bounded by a closed
walk $W'_0$ of length at most $t$.  By the minimality of $|W_0|$, it follows that $|W'_0|=t$.

We now proceed as follows.  Set $W=W'_0$ and $\Lambda=\Lambda_0$, and remove in the data structure all vertices and edges drawn
in $\Lambda$ so that the data structure
now represents $G-\Lambda$.  We process all remaining vertices not contained in $W$ in order, and for each of them test in constant time
whether there exists a closed walk of length $t$ passing through it and homotopically equivalent to the boundary of $\Lambda'$.
Whenever we find such a walk, we repeat the procedure of the previous paragraph and replace $W$ and $\Lambda$
by the obtained walk and the open disk bounded by it, and remove the vertices and edges so that the data structure represents
$G-\Lambda$.

At the end of this procedure, we end up with a closed walk $W$ and an open disk $\Lambda$ satisfying the conclusions of the lemma
(if $\Lambda=f$, then $\{f\}$ is $k$-free).
\end{proof}

Let us remark that in the setting of Lemma~\ref{lemma-testfree1}, if $\{f\}$ is not $k$-free, then the reduction of Corollary~\ref{cor-surfdisk} applies
($\{f\}$ is $(|W|-2)$-free since $|W|$ is minimal).
Another subroutine is used to search for a binding walk around several faces.

\begin{lemma}\label{lemma-testfree2}
For all integers $k,b\ge 0$ and a surface $\Sigma$ other than the sphere, there exists a linear-time algorithm as follows.
Let $G$ be a triangle-free graph with a $2$-cell embedding in $\Sigma$ such that every connected essential
subgraph of $G$ has more than $k$ edges, and let $S$ be a set of faces of $G$ such that
$|S|\le b$.  The algorithm decides whether $S$ is $k$-free, and if not, returns a shortest closed walk $W$
that binds $S$.
\end{lemma}
\begin{proof}
We can assume that $S$ is non-empty, as otherwise it is trivially $k$-free.
Using the algorithm of Lemma~\ref{lemma-testfree1}, we test whether single-element subsets of $S$ are $k$-free,
and record the obtained shortest walks binding them (if any).

Next, we perform the following steps for every set $D\subseteq S$ such that $|D|\ge 2$.
Let $t\le k$ be the largest integer such that $\sum_{f\in D} w_0(f)\ge s(t)$.
Let $\Sigma'$ be the surface obtained from $\Sigma$ by drilling a hole inside each face of $D$.
Let $\Lambda$ be an open disk in $\Sigma$ containing $\Sigma\setminus\Sigma'$.
Build the data structure of Lemma~\ref{lemma-dataemb} for $G$ with $d=t$,
and determine whether there exists a closed
walk $W_0$ of length at most $t$ that is homotopically equivalent in $\Sigma'$ to the boundary of $\Lambda$.

In case such a walk $W_0$ exist, choose a shortest one.
Let $H_{W_0}$ be the subgraph of $G$ consisting of the vertices and edges of $W_0$ and
let $T$ be the set of faces of $H_{W_0}$ intersecting $\Sigma\setminus\Sigma'$.  By Lemma~\ref{lemma-simpc}, we have
$$\sum_{f\in D} w_0(f)\ge s(t)\ge s\left(\sum_{h\in T} |h|\right)\ge \sum_{h\in T} s(h),$$
and either $|T|=1$ or the last inequality is strict.
Hence, there exists $h\in T$ bounded by a closed walk $W$ such that $\sum_{f\in D, f\subseteq h} w_0(f)\ge w(h)$,
and either the inequality is strict, or $\bigcup D\subseteq h$.  Therefore, $W$ binds $S$, and we record $W$.

In the end, we return a shortest recorded walk that binds $S$ (or that $S$ is $k$-free if no walk was recorded).
\end{proof}

Again, if $S$ is not $k$-free, then the minimality of $|W|$ ensures that in the situation of Lemma~\ref{lemma-testfree2}, $S$ is $(|W|-2)$-free.
We can now describe the reduction algorithm.

\begin{lemma}\label{lemma-liberate}
For an integer $k\ge 4$, a rational number $r\ge 0$ and a surface $\Sigma$ other than the sphere, there exists a linear-time algorithm as follows.
Let $G$ be a triangle-free graph with a $2$-cell embedding in $\Sigma$
such that every connected essential subgraph of $G$ has more than $k$ edges.  The algorithm returns a subgraph $G'\subseteq G$
whose embedding in $\Sigma$ induced by the embedding of $G$ is $2$-cell such that
\begin{itemize}
\item every boundary cycle of $G$ belongs to $G'$,
\item every $3$-coloring of $G'$ extends to a $3$-coloring of $G$, and
\item either $w_0(G')\le r$, or $G'$ contains a $k$-free set of faces $S$ such that
$$\sum_{f\in S} w_0(f)>r.$$
\end{itemize}
\end{lemma}
\begin{proof}
Let $b=\lceil r/s(5)\rceil+1$.  Note that there exist only finitely many multisets $M$ of integers greater than $4$
such that $\sum_{t\in M} s(t)\le r$.  Let $m_0$ denote the number of such multisets.

We construct a sequence $(G_0,S_0)$, $(G_1,S_1)$, \ldots, $(G_m,S_m)$
such that for $i=0,\ldots, m$, $G_i$ is a subgraph of $G$, $S_i$ is a set of faces of $G_i$, and
\begin{itemize}
\item the embedding of $G_i$ in $\Sigma$ is $2$-cell and every boundary cycle of $G$ belongs to $G_i$,
\item every $3$-coloring of $G_i$ extends to a $3$-coloring of $G$,
\item $|S_i|\le b$, and
\item for every $f\in S_i$, $|f|\ge 5$ and the set $\{f\}$ is $k$-free in $G_i$.
\end{itemize}
Let $G_0$ be the subgraph of $G$ obtained by applying the algorithm of Lemma~\ref{lemma-elim24} and suppressing
all faces of length two, and let $S_0=\emptyset$.  Assuming we already constructed $(G_i,S_i)$ for some $i\ge 0$, we proceed as follows.

We apply the algorithm of Lemma~\ref{lemma-testfree2} to test whether $S_i$ is $k$-free in $G_i$.  If not, let $W_0$ be the closed
walk returned by the algorithm that binds $S_i$ with respect to an open disk $\Lambda_0\subset\Sigma$.
By Corollary~\ref{cor-surfdisk}, every $3$-coloring of $G_i-\Lambda_0$ extends to a $3$-coloring of $G_i$, and thus also to a $3$-coloring of $G$.
Apply the algorithm of Lemma~\ref{lemma-testfree1} to $G_i-\Lambda_0$ and its face $\Lambda_0$, and
let $W$ be the closed walk returned by the algorithm that binds $\{\Lambda_0\}$ with respect to an open disk $\Lambda\subset\Sigma$,
or set $W=W_0$ and $\Lambda=\Lambda_0$ when $\{\Lambda_0\}$ already is $k$-free.
Let $G_{i+1}=G_i-\Lambda$ and let $S_{i+1}$ be obtained from $S_i$ by removing the faces contained in $\Lambda$ and by adding $\Lambda$.
Note that by Corollary~\ref{cor-surfdisk}, every $3$-coloring of $G_{i+1}$ extends to a $3$-coloring of $G_i-\Lambda_0$, and thus
also to a $3$-coloring of $G$.  Furthermore, $|S_{i+1}|\le |S_i|\le b$.

Hence, assume that $S_i$ is $k$-free in $G_i$.  If either $\sum_{f\in S_i} w_0(f)>r$ or $S_i$ contains all faces of $G_i$ of length at least $5$, then we
set $m=i$ and end the procedure.
Otherwise, note that $|S_i|\le r/s(5)\le b-1$, and let $f\not\in S$ be a face of $G_i$ of length at least $5$.
Let $\Lambda\subset\Sigma$ be an open disk found using Lemma~\ref{lemma-testfree1} applied for $f$ such that
$f\subseteq \Lambda$ and $\Lambda$ is $k$-free in $G_i-\Lambda$.  By the choice of $G_0$, the boundary walk of $\Lambda$ has
length at least $5$.  Let $G_{i+1}=G_i-\Lambda$ and let $S_{i+1}$ be obtained from $S_i$ by removing the faces contained in $\Lambda$ and by adding $\Lambda$.
Note that by Corollary~\ref{cor-surfdisk}, every $3$-coloring of $G_{i+1}$ extends to a $3$-coloring of $G_i$, and thus also
to a $3$-coloring of $G$.

This finishes the description of the construction of the sequence $G_0$, $G_1$, \ldots, $G_m$.
Let us now give a bound on the length of the sequence.
For $i\ge 0$, let $M_i$ denote the sequence of the lengths of faces in $S_i$ in the non-increasing order.
Observe that in order to obtain $S_{i+1}$, some subset $X$ of faces in $S_i$ is replaced by another face $\Lambda$, and
every face in $X$ is strictly shorter than $\Lambda$ by the assumption that $\{f\}$ is $k$-free in $G_i$ for every $f\in S_i$.
Hence, $M_{i+1}$ is lexicographically strictly larger than $M_i$.  Consequently, $M_i\neq M_j$ for every $i\neq j$,
and thus there exist at most $m_0$ values of $i$ such that $\sum_{f\in S_i} w_0(f)\le r$.
Furthermore, observe that if $S_i$ is not $k$-free, then $|S_{i+1}|<|S_i|$, and thus there do not exist $b$ consecutive
values of $i$ such that $S_i$ is not $k$-free.  We conclude that the described algorithm terminates after $m\le b(m_0+1)$ steps,
which is a constant depending only on $r$.

By the construction, either $S_m$ is $k$-free in $G_m$ and $\sum_{f\in S_m} w_0(f)>r$, or $w_0(G_m)\le r$.  Therefore,
we can set $G'=G_m$.
\end{proof}

As a corollary, we obtain the special case of Theorem~\ref{thm-mainalg} when $\Sigma$ is the disk.

\begin{corollary}\label{cor-mainalg-disk}
For every integer $n$, there exists a linear-time algorithm with input
\begin{itemize}
\item $G$: a triangle-free graph embedded in the disk $\Delta$ with the boundary cycle $B$ of length at most $n$, and
\item $\psi$: a $3$-coloring of $B$,
\end{itemize}
which either correctly decides that there exists a $3$-coloring $\varphi$ of $G$ extending $\psi$,
or returns a subgraph $G'\subseteq G$ such that $B\subseteq G'$, $w_\eta(G')\le s(|B|-2)$ and $\psi$ does not extend to a $3$-coloring of $G'$.
\end{corollary}
\begin{proof}
Let $G'$ be the subgraph of $G$ obtained by applying the algorithm of Lemma~\ref{lemma-liberate} with $k=|B|-2$ and $r=s(|B|-2)$.
Note that $\psi$ extends to a $3$-coloring of $G$ if and only if it extends to a $3$-coloring of $G'$.
If $w_0(G')\le r=s(|B|-2)$, then we can decide whether $\psi$ extends to a $3$-coloring of $G'$ using the algorithm of Theorem~\ref{thm-quadalg},
by testing all the possible colorings of the vertices incident with faces of $G'$ of length greater than $4$ (there are at most $5s(|B|-2)/s(5)$ such vertices).
Otherwise, $G'$ contains a $(|B|-2)$-free set of total weight greater than $s(|B|-2)$, and Lemma~\ref{lemma-disk} implies that
$\psi$ extends to a $3$-coloring of $G'$.
\end{proof}

We often use the following consequence.
\begin{corollary}\label{cor-mainalg-cuttodisk}
For every surface $\Sigma$ and integer $n$, there exists a linear-time algorithm with input
\begin{itemize}
\item $G$: a triangle-free graph embedded in $\Sigma$,
\item $\psi$: a $3$-coloring of a subgraph $B$ of $G$,
\item $Q$: a subgraph of $G$ with at most $n$ vertices, such that $B\subseteq Q$ and every face of $Q$
containing vertices or edges of $G$ is $2$-cell,
\end{itemize}
which correctly decides whether there exists a $3$-coloring $\varphi$ of $G$ extending $\psi$.
\end{corollary}
\begin{proof}
The algorithm iterates over all $3$-colorings of $Q$ that extend $\psi$, and for each of them checks whether it extends to $G$.
To do so, it suffices to check for each face $f$ of $Q$ containing vertices or edges of $G$ (which is $2$-cell by the assumptions)
whether the corresponding precoloring of the boundary cycle of $G_f$ extends to a $3$-coloring of $G_f$;
this can be done using the algorithm from Corollary~\ref{cor-mainalg-disk}.
\end{proof}

\section{Linear-time decision algorithm}\label{sec-linalg}

We need a few more concepts to generalize the algorithm to all surfaces.
Suppose that a graph $G$ is embedded in a surface $\Sigma$ so that every cuff of $\Sigma$ traces a cycle in $G$, and
let $B_1$, \ldots, $B_m$ be the boundary cycles of $G$.  We say that a $3$-coloring $\psi$ of the boundary cycles of $G$
is \emph{locally blocked} of there exists $i\in \{1,\ldots,m\}$ and a simple closed curve $c$ in $\Sigma$ homotopically equivalent to $B_i$
such that the subgraph of $G$ consisting of the vertices and edges fully drawn between $c$ and $B_i$ (inclusive)
has no $3$-coloring matching $\psi$ on $B_i$.  Conversely, we say that a boundary cycle $B_i$ is \emph{$\psi$-irrelevant}
if $G$ has no $3$-coloring matching $\psi$ on all boundary cycles distinct from $B_i$.
We say that $G$ is \emph{boundary-linked} if $\Sigma$ has non-zero genus and
for $i=1,\ldots,m$, every cycle surrounding the cuff of $B_i$ has length at least $|B_i|$.
We now argue that boundary-linkedness prevents any coloring from being locally blocked.

\begin{lemma}\label{lemma-nonblock}
Let $G$ be a triangle-free graph embedded in a surface $\Sigma$ of non-zero genus
such that every cuff of $\Sigma$ traces a cycle in $G$.  If $G$ is boundary-linked, then no $3$-coloring of the boundary cycles
of $G$ is locally blocked.
\end{lemma}
\begin{proof}
Consider any simple closed curve $c$ in $G$ homotopically equivalent to a boundary cycle $C$ and a $3$-coloring $\psi$ of $C$.
Let $G_c$ be the subgraph of $G$ drawn between $C$ and $c$, and consider the drawing of $G_c$
in the disk obtained by cutting $\Sigma$ along $c$ and patching the hole bounded by $c$; let $f$ be the resulting face of $G_c$
containing the patch.  Since $G$ is boundary linked, $f$ has length at least $|C|$, and thus $w_0(f)\ge s(|C|)>s(|C|-2)$.
Moreover, the set $\{f\}$ is $(|C|-2)$-free. Hence, $\psi$ extends to a $3$-coloring of $G_c$ by Lemma~\ref{lemma-disk}.
\end{proof}

We now give a generalization of Lemma~\ref{lemma-disk}.
\begin{lemma}\label{lemma-gen}
Let $G$ be a triangle-free graph embedded in a surface $\Sigma$
such that every cuff of $\Sigma$ traces a cycle in $G$, and let $B$ denote the union of the boundary cycles of $G$.  Suppose that
every $4$-cycle in $G$ is contractible and every connected essential subgraph of $G$
has more than $k=\lceil w_\eta(B)\rceil$ edges.  Let $\psi$ be a $3$-coloring of $B$ that is not locally blocked.
If $G$ contains a $k$-free set $S$ of faces such that $\sum_{f\in S} w_0(f)>w_\eta(B)$,
then either $\psi$ extends to a $3$-coloring of $G$, or some boundary cycle of $G$ is $\psi$-irrelevant.
\end{lemma}
\begin{proof}
By Gr\"otzsch's theorem and Lemma~\ref{lemma-disk}, we can assume $\Sigma$ is neither the sphere nor the disk.

Suppose $\psi$ does not extend to a $3$-coloring of $G$.
By Theorem~\ref{thm-almstruct},
$G$ has a subgraph $H$ such that $B\subseteq H$, $w_\eta(H)\le w_\eta(B)$, and
$\psi$ does not extend to a $3$-coloring of $H$.
We can without loss of generality assume that every component of $H$ contains a non-contractible cycle
(components only containing contractible cycles are disjoint from the boundary and $3$-colorable by Gr\"otzsch's theorem,
and omitting them can only decrease $w_\eta(H)$).

Since $w_\eta(H)\le w_\eta(B)\le k$, every face $h$ of $H$ has length at most $k$.
If the embedding of $H$ were $2$-cell, then $\sum_{f\in S, f\subseteq h} w_0(f)\le w_0(h)$ would hold for every face $h$ of $H$,
since $S$ is $k$-free.  However, that would imply $\sum_{f\in S} w_0(f)\le w_0(H)\le w_\eta(B)$, contradicting the assumptions.

Therefore, $H$ contains a face $h$ that is not $2$-cell.  Consider any boundary walk $W$ of $h$.
Since $h$ has length at most $k$ and $G$ does not contain any connected essential subgraph with at most $k$ edges,
$W$ is either contractible or homotopically equivalent to a boundary cycle $C$.

Let $c$ be a simple closed curve drawn along $W$ inside the face $h$.
If $c$ were contractible, then the closed disk $\Delta\subseteq\Sigma$ bounded by $c$ would be disjoint from the drawing of $H$
(since every component of $H$ contains a non-contractible cycle), and thus the face $h$ would be $2$-cell.
Hence, $c$ is homotopically equivalent to a boundary cycle $C$.  Since $\psi$ is not locally blocked, it extends to a $3$-coloring
of the part $H_c$ of $H$ drawn between $C$ and $c$.  Since $\psi$ does not extend to a $3$-coloring of $H$, it follows that
it does not extend to a $3$-coloring of $H-V(H_c)$.  Since $H\subseteq G$, $\psi$ does not extend to a $3$-coloring of
$G-V(H_c)\subseteq G-V(C)$.  Hence, the boundary cycle $C$ is $\psi$-irrelevant.
\end{proof}

We can now give the algorithm for the special case where the input graph is boundary-linked,
has no non-contractible $4$-cycles, and no small connected essential subgraphs.

\begin{theorem}\label{thm-mainalg-spec}
For every surface $\Sigma$ and integer $n$, there exists a linear-time algorithm with input
\begin{itemize}
\item $G$: a triangle-free graph $2$-cell embedded in $\Sigma$ with boundary cycles $B$ of total length at most $n$, such that
\begin{itemize}
\item every $4$-cycle in $G$ is contractible,
\item every connected essential subgraph of $G$ has more than $\max(n, \lceil w_\eta(B)\rceil)$ edges, and
\item either $\Sigma$ has genus zero or $G$ is boundary-linked; and,
\end{itemize}
\item $\psi$: a $3$-coloring of $B$,
\end{itemize}
which either correctly decides that there exists a $3$-coloring $\varphi$ of $G$ extending $\psi$, or
returns a subgraph $G_0\subseteq G$ such that $B\subseteq G_0$, $w_\eta(G_0)\le w_\eta(B)$ and $\psi$ does not extend to a $3$-coloring of $G_0$.
\end{theorem}
\begin{proof}
We proceed by induction on the number of cuffs of $\Sigma$; hence, we can assume that the algorithm exists for all
surfaces obtained from $\Sigma$ by patching at least one cuff.  Furthermore, by Corollary~\ref{cor-mainalg-disk},
we can assume $\Sigma$ is not the sphere with at most one hole.

Apply the algorithm of Lemma~\ref{lemma-liberate} with $k=\lceil w_\eta(B)\rceil$ and $r=w_\eta(B)$, and let $G'$ be the resulting subgraph
such that every $3$-coloring of $G'$ extends to a $3$-coloring of $G$.  It suffices to decide whether $\psi$ extends to a $3$-coloring of $G'$.
Note that every $4$-cycle in $G'$ is contractible, every connected essential subgraph of $G'$ has more than $\max(n, k)$ edges, and
if $\Sigma$ has positive genus, then $G'$ is boundary-linked.

If $w_0(G')\le r$, we can decide whether $\psi$ extends to a $3$-coloring of $G'$ using the algorithm of Theorem~\ref{thm-quadalg},
by testing all the possible colorings of the vertices incident with faces of $G'$ of length greater than $4$ (there are at most $5w_\eta(B)/s(5)$ such vertices).
If it does, then $\psi$ also extends to a $3$-coloring of $G$.  If $\psi$ does not extend to a $3$-coloring of $G'$, then we can set $G_0=G'$.
Therefore, suppose that $w_0(G')>r$, and thus $G'$ contains a $k$-free set of faces $S$ such that $\sum_{f\in S} w_0(f)>r$.

Next, for each boundary cycle $C$, we recursively call the current algorithm for the embedding of $G'$ in the surface $\Sigma'$ obtained from
$\Sigma$ by patching the cuff of $C$ and for the restriction $\psi_{\bar{C}}$ of $\psi$ to $B-V(C)$.
Note that patching does not change the genus of the surface, cannot turn a non-essential subgraph into an essential one,
and if a cycle surrounds a cuff in $\Sigma'$ but not in $\Sigma$, then it is essential in $\Sigma$, and thus the assumptions
of the theorem are satisfied by the embedding of $G'$ in $\Sigma'$.
If $\psi_{\bar{C}}$ does not extend to a $3$-coloring of $G'$, we obtain a subgraph $G'_0\subseteq G'$ with $w_\eta(G'_0)\le w_\eta(B-V(C))$
such that $\psi_{\bar{C}}$ does not extend to a $3$-coloring of $G'_0$.  Then, it suffices to set $G_0=G'_0\cup C$
and observe that $w_\eta(G_0)\le w_\eta(B)$ and that $\psi$ does not extend to a $3$-coloring of $G_0$.

Hence, we can assume that for every boundary cycle $C$, the restriction of $\psi$ to $B-V(C)$ extends to a $3$-coloring of $G'$,
and thus no boundary cycle of $G'$ is $\psi$-irrelevant.  Recall that if $\Sigma$ has genus zero, then it has at least two holes,
and thus since no boundary cycle of $G'$ is $\psi$-irrelevant, $\psi$ is not locally blocked.  If $\Sigma$ has non-zero genus, then by the assumptions $G'$ is boundary-linked,
and thus by Lemma~\ref{lemma-nonblock}, we again conclude that $\psi$ is not locally blocked.

Therefore, $\psi$ extends to a $3$-coloring of $G'$ (and thus also $G$) by Lemma~\ref{lemma-gen}.
\end{proof}

We can allow the boundary cycles to have length $4$ by a minor modification to the algorithm of Theorem~\ref{thm-mainalg-spec}: Before
running the algorithm,
subdivide an edge in each cuff of length $4$ by a vertex and extend $\psi$ by giving the new vertex a color distinct from the colors of
its neighbors (this does not affect boundary-linkedness, since we assume absence of non-contractible non-boundary $4$-cycles).

\begin{corollary}\label{cor-mainalg-spec}
For every surface $\Sigma$ and integer $n$, there exists an integer $N_{\Sigma,n}$ and a linear-time algorithm with input
\begin{itemize}
\item $G$: a triangle-free graph $2$-cell embedded in $\Sigma$ with boundary cycles $B$ of total length at most $n$, such that
\begin{itemize}
\item every non-boundary $4$-cycle in $G$ is contractible,
\item every connected essential subgraph of $G$ has more than $N_{\Sigma,n}$ edges, and
\item either $\Sigma$ has genus zero or $G$ is boundary-linked; and,
\end{itemize}
\item $\psi$: a $3$-coloring of $B$,
\end{itemize}
which correctly decides whether there exists a $3$-coloring $\varphi$ of $G$ extending $\psi$.
\end{corollary}

Next, we deal with other non-contractible $4$-cycles as well as the boundary-linkedness assumption.
\begin{theorem}\label{thm-mainalg-spec4}
For every surface $\Sigma$ and every integer $n\ge 2$, letting $N_{\Sigma,n}$ be as in Corollary~\ref{cor-mainalg-spec},
there exists a linear-time algorithm with input
\begin{itemize}
\item $G$: a triangle-free graph $2$-cell embedded in $\Sigma$ with boundary cycles $B$ of total length at most $n$, such that
every connected essential subgraph of $G$ has more than $N_{\Sigma,n}+2n$ edges, and
\item $\psi$: a $3$-coloring of $B$,
\end{itemize}
which correctly decides whether there exists a $3$-coloring $\varphi$ of $G$ extending $\psi$.  Furthermore, if $\Sigma$ is the cylinder
(the sphere with two holes), then we can omit the restriction on the connected essential subgraphs.
\end{theorem}
\begin{proof}
By Corollary~\ref{cor-mainalg-disk}, we can assume that $\Sigma$ has non-zero genus or at least two cuffs.

Suppose first that $\Sigma$ is the cylinder.  Apply the algorithm of Lemma~\ref{lemma-splitcyl} with $d=k$, and let $C_0$, $C_1$, \ldots, $C_m$
be the resulting cycles.  For $i=1,\ldots, m$, we can decide which $3$-colorings of $C_{i-1}\cup C_i$ extend to the subgraph $G_i$ of $G$ drawn between $C_{i-1}\cup C_i$,
as follows.  If the distance between $C_{i-1}$ and $C_i$ is at most $N_{\Sigma,2n}$, then let $P$ be a shortest path between $C_{i-1}$ and $C_i$.
Using Corollary~\ref{cor-mainalg-cuttodisk} to $G_i$ with $Q=C_{i-1}\cup P\cup C_i$, we can decide which
$3$-colorings of $C_{i-1}\cup C_i$ extend to a $3$-coloring of $G_i$.
If the distance between $C_{i-1}$ and $C_i$ is greater than $N_{\Sigma,2n}$, then since $\Sigma$ is a cylinder,
we conclude that $G_i$ does not contain any connected essential subgraph with at most $N_{\Sigma,2n}$ vertices, and thus
which $3$-colorings of $C_{i-1}\cup C_i$ extend to $G_i$ can be decided using the algorithm of Corollary~\ref{cor-mainalg-spec}.

Finally, we can combine the information by a straightforward dynamic programming to determine whether $\psi$ extends to a $3$-coloring of $G$:
For $i=0,\ldots,m$, we compute the set $\Psi_i$ of colorings of $C_0\cup C_i$ that extend to a $3$-coloring of the subgraph of $G$
between $C_0$ and $C_i$, and then we test whether $\psi\in\Psi_m$.  To determine whether a coloring $\theta$ belongs to $\Psi_i$ for
$i\ge 1$, it suffices to check whether there exists a coloring $\theta'\in\Psi_{i-1}$ such that $\theta'\restriction V(C_0)=\theta\restriction V(C_0)$
and the restriction of $\theta\cup\theta'$ to $C_{i-1}\cup C_i$ extends to a $3$-coloring of $G_i$.

Hence, we can assume that $\Sigma$ is not the cylinder.  Note that every non-contractible 4-cycle in $G$ surrounds a cuff,
since $G$ does not contain connected essential subgraphs with $4$ edges.
For each cuff $C$ of $\Sigma$, let $\Sigma_C$ be the surface obtained from $\Sigma$ by patching $C$, and let $f_C$ be the face of $G$ bounded by $C$
in its embedding in $\Sigma_C$.  Let $W_C$ be the closed walk obtained by applying the algorithm of Lemma~\ref{lemma-testfree1} to $\{f_C\}$ with $k=n$, and
let $\Lambda_C$ be the open disk bounded by $W_C$.  Note that for any distinct cuffs $C_1$ and $C_2$, the closures of $\Lambda_{C_1}$ and $\Lambda_{C_2}$
are disjoint, since every connected essential subgraph of $G$ has more than $2n$ edges, and since $\Sigma$ is not the cylinder.

Let $C_1$, \ldots, $C_t$ be the cuffs of $\Sigma$, let $\Sigma'=\Sigma\setminus(\Lambda_{C_1}\cup\ldots\cup \Lambda_{C_t})$ and let $G'$ be the subgraph of
$G$ embedded in $\Sigma'$.  Note that $\Sigma'$ is homeomorphic to $\Sigma$, that $G'$ does not contain any non-contractible $4$-cycles other than
the boundary cycles, that if $\Sigma$ has positive genus, then $G'$ is boundary-linked, and that $G'$ does not contain any connected essential subgraphs with at most $N_{\Sigma,n}$ edges.
By Corollary~\ref{cor-mainalg-spec}, we can decide which $3$-colorings of the boundary cycles $W_1\cup \ldots\cup W_t$ of $G'$ extend to a $3$-coloring of $G'$.
For $i=1,\ldots, t$, we can decide which $3$-colorings of $C_i\cup W_i$ extend to the subgraph of $G$ drawn between $C_i$ and $W_i$,
by the cylinder case when $W_i$ and $C_i$ are vertex-disjoint, and by Corollary~\ref{cor-mainalg-cuttodisk} otherwise.
By combining this information, we can decide whether $\psi$ extends to a $3$-coloring of $G$.
\end{proof}

To give the full algorithm, it now suffices to deal with the essential subgraphs.

\begin{proof}[Proof of Theorem~\ref{thm-mainalg}]
Without loss of generality, we can assume that the embedding of $G$ in $\Sigma$ is $2$-cell.

Let $\nu(\Pi,n)=N_{\Pi,n}+2n$ for every surface $\Pi$ and integer $n$, where $N_{\Pi,n}$ is as in Corollary~\ref{cor-mainalg-spec}.  Apply the algorithm of Lemma~\ref{lemma-qstruc} to $G$, obtaining a subgraph $H$ of $G$.
For every face $h$ of $H$, determine which $3$-colorings of the boundary of $h$ extend to a $3$-coloring of $G_h$ by applying the algorithm of Theorem~\ref{thm-mainalg-spec4}.
By combining this information, we can determine which $3$-colorings of $H$ extend to a $3$-coloring of $G$, and thus also whether $\psi$ extends to a $3$-coloring of $G$.
\end{proof}

Let us remark that in Theorem~\ref{thm-mainalg-spec} (as well as in all the algorithms described before Theorem~\ref{thm-mainalg-spec}),
in case the precoloring did not extend to a $3$-coloring
of the whole graph, we were able to provide a certificate for this fact---a ``near-quadrangulated'' subgraph (of bounded $w_\eta$-weight)
to which the $3$-coloring does not extend.  In the algorithm of Theorem~\ref{thm-mainalg-spec4}, we use dynamic programming,
making the structure of such a certificate less clear.  Moreover, an inspection of~\cite{trfree6} shows that for a similar reason,
it is not straightforward to certify that the precoloring does not extend to the near-quadrangulated subgraph, either.  With a substantial additional
work, these issues can be dealt with; Dvo\v{r}\'ak and Lidick\'{y}~\cite{cylgen-part3} show that the non-extendability
of a precoloring of the boundary cycles in a triangle-free graph embedded in a surface $\Sigma$ can be certified by
a subgraph that has a description of size bounded by a constant depending only on $\Sigma$ and the number of precolored vertices.

\section{Finding a $3$-coloring}\label{sec-repalg}

Note that the arguments used in the proof of Theorem~\ref{thm-mainalg} (e.g., the proof of Lemma~\ref{lemma-disk})
do not give a way to find the $3$-coloring in case one exists.  Hence, we need additional ideas to get such a $3$-coloring.

The basic case we need to consider is that of graphs of girth at least $5$.  It is possible to algorithmize our arguments
from the previous papers of this series~\cite{trfree2,trfree3} that deal with this case, but it is not entirely straightforward.
Fortunately, Postle~\cite{postledistr-focs} came up with a more elegant argument.  The following presentation is a bit simpler
than Postle's, as his one adds several additional ideas to obtain an efficient distributed algorithm.

Let $G$ be a graph drawn in a surface $\Sigma$ and let $F$ be an induced subgraph of $G$.  Let $\partial F$ denote the set of
vertices $v\in V(F)$ such that either $v$ is drawn in the boundary of $\Sigma$, or $v$ is incident with an edge of $E(G)\setminus E(F)$.
We say that a vertex $u\in V(F)\setminus\partial F$ is \emph{$F$-irrelevant} if every $3$-coloring of $\partial F$ that extends to a $3$-coloring
of $F-u$ also extends to a $3$-coloring of $F$.  We say that $F$ is \emph{reducible} if it contains an $F$-irrelevant vertex.
\begin{observation}\label{obs-redu}
Let $G$ be a graph drawn in a surface, let $\psi$ be a function from boundary vertices of $G$ to $\{1,2,3\}$,
and let $F$ be an induced subgraph of $G$.  If $u$ is an $F$-irrelevant vertex, then given a $3$-coloring $\varphi'$ of $G-u$,
we can obtain a $3$-coloring $\varphi$ of $G$ by changing the colors of vertices of $V(F)\setminus (\{u\}\cup \partial F)$ and assigning a color to $u$.
In particular, if $\varphi'$ extends $\psi$, then so does $\varphi$.
\end{observation}

The key insight is that if $F$ is sufficiently large compared to $\partial F$ and has girth at least five, then it is reducible.
\begin{lemma}\label{lemma-redu}
Let $G$ be a graph of girth at least five embedded in a surface $\Sigma$ of genus $g$ and let $F$ be an induced subgraph of $G$.
If $|V(F)|>\tfrac{5}{s(5)}((48\eta+5)|\partial F|+120\eta g)$, then $F$ is reducible.
\end{lemma}
\begin{proof}
Let $\Sigma'$ be the surface obtained by patching all cuffs of $\Sigma$, and then for each vertex $v\in\partial F$, drilling a hole
next to $v$ so that the drawing of $F$ intersects the resulting cuff exactly in $v$.  Let $F'$ be obtained from $F$ by,
for each vertex $v\in \partial F$, adding a 5-cycle $B_v$ tracing the incident cuff of $\Sigma'$, and let $B=\bigcup_{v\in\partial F} B_v$.
By Theorem~\ref{thm-almstruct}, there exists $H\subseteq F'$ such that $B\subseteq H$,
$$w_\eta(H)\le w_\eta(B)\le (48\eta+5)|\partial F|+120\eta g,$$
and every $3$-coloring of $B$ that extends to a $3$-coloring of $H$ also extends to a $3$-coloring of $F'$.
Since $F'$ has girth at least five, $w_\eta(h)\ge \frac{s(5)}{5}|h|$ holds for every face $h$ of $H$,
and thus $|V(H)|\le \frac{5}{s(5)}w_\eta(H)$.  Hence, we have $|V(F)|>|V(H)|$.  Every vertex in $V(F)\setminus V(H)$
is $F$-irrelevant, and thus $F$ is reducible.
\end{proof}

Let us recall the well-known result that graphs on surfaces have sublinear balanced separators.
A \emph{separator} in a graph $G$ is a pair $(A,B)$ of induced subgraphs such that $G=A\cup B$.
The \emph{order} of the separator is $|V(A\cap B)|$, and the separator is \emph{balanced} if
$|V(A)\setminus V(B)|\le \tfrac{2}{3}|V(G)|$ and $|V(B)\setminus V(A)|\le \tfrac{2}{3}|V(G)|$.
\begin{theorem}[Gilbert et al.~\cite{gilbert}]\label{thm-balsep}
For every $g\ge 0$, there exists a constant $\sigma_g=O(\sqrt{g})$ such that every graph $G$ drawn
in a surface of genus $g$ has a balanced separator of order at most $\sigma_g\sqrt{|V(G)|}$.
\end{theorem}

We can use this result to obtain an induced subgraph with a small boundary as follows.
\begin{lemma}\label{lemma-decr}
Let $G$ be a graph drawn in a surface $\Sigma$ of genus $g$, and let $F$ be an induced subgraph of $G$.
Let $a$ and $b$ be positive real numbers and let $c=6a\sigma_g+b/\sigma_g$.
If $|V(F)|\ge a|\partial F|+b-c\sqrt{|V(F)|}$ and $|V(F)|>9\sigma^2_g$, then there exists a proper induced subgraph
$F'$ of $G$ such that $|V(F)|/3\le |V(F')|<|V(F)|$ and $|V(F')|\ge a|\partial F'|+b-c\sqrt{|V(F')|}$.
\end{lemma}
\begin{proof}
Let $n=|V(F)|$.
For an induced subgraph $H$ of $G$, let us define $e(H)=a|\partial H|+b-|V(H)|$, so that $e(F)\le c\sqrt{n}$.
By Theorem~\ref{thm-balsep}, $F$ has a balanced separator $(A,B)$ of order at most $\sigma_g\sqrt{n}$.
Since $|V(A)|,|V(B)|\ge n/3$, we have
$\sqrt{|V(A)|}+\sqrt{|V(B)|}\ge \bigl(\sqrt{2/3}+\sqrt{1/3}\bigl)\sqrt{n}>\frac{4}{3}\sqrt{n}$.
Moreover, note that $\partial A\subseteq (\partial F\setminus V(B))\cup V(A\cap B)$ and $\partial B\subseteq (\partial F\setminus V(A))\cup V(A\cap B)$.
Hence,
\begin{align*}
e(A)+e(B)&=a(|\partial A|+|\partial B|)+2b-(|V(A)|+|V(B)|)\\
&\le a(|\partial F|+2\sigma_g\sqrt{n})+2b-n=e(F)+2a\sigma_g\sqrt{n}+b\\
&\le \bigl(c+2a\sigma_g+\tfrac{b}{\sqrt{n}}\bigr)\sqrt{n}\le \bigl(c+2a\sigma_g+\tfrac{b}{3\sigma_g}\bigr)\sqrt{n}\\
&\le \tfrac{4}{3}c\sqrt{n}\le c(\sqrt{|V(A)|}+\sqrt{|V(B)|}).
\end{align*}
Consequently, there exists $F'\in \{A,B\}$ such that $e(F')\le c\sqrt{|V(F')|}$.
Also, we have $|V(F')|\le \max(|V(A)|,|V(B)|)\le \tfrac{2}{3}n+\sigma_g\sqrt{n}<n$, as required.
\end{proof}

Repeatedly applying Lemma~\ref{lemma-decr}, we obtain the following consequence.
\begin{corollary}\label{cor-decr}
Let $G$ be a graph drawn in a surface $\Sigma$ of genus $g$.  Let $a$ and $b$ be positive real numbers,
let $c=6a\sigma_g+b/\sigma_g$, and let $t\ge 9\sigma^2_g$ be a real number.
If $|V(G)|\ge a|\partial G|+b$ and $|V(G)|>t$, then there exists a proper induced subgraph
$F$ of $G$ such that $t/3<|V(F)|\le t$ and $|V(F)|\ge a|\partial F|+b-c\sqrt{|V(F)|}$.
\end{corollary}
In particular, for $t=\max(9\sigma^2_g,3c^2)$, $t/3<|V(F)|$ implies $c\sqrt{|V(F)|}<|V(F)|$,
an thus $|V(F)|>\tfrac{a}{2}|\partial F|+\tfrac{b}{2}$.  Let us now combine Corollary~\ref{cor-decr} with Lemma~\ref{lemma-redu},
setting $a=\tfrac{10(48\eta+5)}{s(5)}$ and $b=\tfrac{1200\eta g}{s(5)}$.
\begin{corollary}
For every surface $\Sigma$ and integer $n$, there exists a constant $t_{\Sigma,n}$ such that
if $G$ is a graph of girth at least five drawn in $\Sigma$ with at most $n$ boundary vertices and $|V(G)|>t_{\Sigma,n}$,
then $G$ contains a reducible induced subgraph with at most $t_{\Sigma,n}$ vertices.
\end{corollary}

Note that for a fixed $\Sigma$ and $n$, there are only finitely many graphs with at most $t_{\Sigma,n}$ vertices
and with some vertices marked as boundary ones, and for each of them, we can check whether they appear in $G$
in linear time using the algorithm of Eppstein~\cite{bib-eppstein99}.  If $G$ contains such an induced subgraph $F$ which is reducible,
and $u\in V(F)\setminus\partial F$ is $F$-irrelevant, we can solve the problem of finding a $3$-coloring of $G$ that extends a precoloring of
its boundary vertices by recursing on $G-u$ and extending the obtained coloring using Observation~\ref{obs-redu}
(since $|V(F)|$ is bounded by a constant, we can do this by brute force).  Furthermore, once the recursion reaches a graph with at most $t_{\Sigma,n}$
vertices, we can find the $3$-coloring (or decide it does not exist) by brute force.
\begin{corollary}\label{cor-girth5}
For every surface $\Sigma$ and integer $k$, there exists a quadratic-time algorithm with input
\begin{itemize}
\item $G$: a graph of girth at least five embedded in $\Sigma$ with at most $k$ boundary vertices, and
\item $\psi$: a function from boundary vertices of $G$ to $\{1,2,3\}$,
\end{itemize}
which correctly decides whether there exists a $3$-coloring $\varphi$ of $G$ such that $\varphi(v)=\psi(v)$ for every boundary vertex $v$ of $G$,
and outputs such a coloring in the affirmative case.
\end{corollary}
Let us remark that one can improve the complexity of the algorithm from Corollary~\ref{cor-girth5} to linear by either carefully keeping track of
potentially reducible subgraphs or applying further ideas of Postle~\cite{postledistr-focs}.  However, as we are anyway only aiming to
obtain a quadratic-time algorithm for the triangle-free case, we will not complicate the exposition by describing these improvements.

\begin{proof}[Proof of Theorem~\ref{thm-mainalgrep}]
We construct the algorithm by induction on the complexity of the surface.  Without loss of generality, we can assume that every boundary
cycle has length at least five (otherwise, subdivide the shorter boundary cycles and color the resulting vertices of degree two
arbitrarily).  First, we use the algorithm of Theorem~\ref{thm-mainalg}
to determine whether $\psi$ extends to a $3$-coloring of $G$.  If it does not, we output this answer and end.  Hence, suppose that $\psi$ extends to a $3$-coloring of $G$.

In linear time, we decide whether $G$ contains a non-contractible $(\le\!5)$-cycle $K$ distinct from the boundary cycles.
If that is the case, we let $H=B\cup K$ and we use the algorithm of Theorem~\ref{thm-mainalg} to find a $3$-coloring $\psi'$ of $H$ that
extends to a $3$-coloring of $G$ and whose restriction to $B$ is equal to $\psi$.  For every face $h$ of $H$, we determine
a $3$-coloring of $G_h$ that extends $\psi'$ by a recursive call.  By combining these colorings, we obtain a $3$-coloring of $G$
that extends $\psi$.  Hence, we can assume that $G$ contains no non-contractible $(\le\!5)$-cycle distinct from the boundary cycles.
Similarly, we can assume that every boundary cycle is induced,
and that the distance between any two boundary cycles is at least three.

Next, in linear time we decide whether $G$ contains a contractible $(\le\!5)$-cycle $K$
that does not bound a face.  If such a cycle $K$ exists, then let $\Lambda$ be the open disk bounded by $K$.  We first extend $\psi$ to a $3$-coloring of $G-\Lambda$ by
a recursive call, and then extend the coloring to $G_\Lambda$ using the algorithm of Theorem~\ref{thm-ext45}.
Hence, we can assume that every contractible $(\le\!5)$-cycle in $G$ bounds a face.  Similarly, we can assume that every vertex of $G$ of degree at most two belongs to $B$.

If $G$ does not contain any $4$-face, then $G$ has girth at least five, and
we can apply the algorithm of Corollary~\ref{cor-girth5}.  Hence, we can assume $G$ has a $4$-face $f=v_1v_2v_3v_4$.  Since the boundary cycles of $G$ are induced and the distance between them
is at least three, we can assume that $v_2\not\in V(B)$.  Suppose that $v_1, v_3\in V(B)$.  Then both $v_1$ and $v_3$ belong to the same boundary cycle $C$, and since $G$ is triangle-free,
the distance between $v_1$ and $v_3$ in $C$ is at least two.  Hence, if $h$ is a face of $H=B+v_1v_2v_3$, then either $\Sigma_h$ is strictly less complex than $h$, or $\Sigma_h$ is homeomorphic to
$\Sigma$ and all the boundary cycles of $G_h$ are at most as long as the corresponding boundary cycles of $G$.  We use the algorithm of Theorem~\ref{thm-mainalg} to find a $3$-coloring $\psi'$ of $H$ that
extends to a $3$-coloring of $G$ and whose restriction to $B$ is equal to $\psi$, and extend $\psi'$ to a $3$-coloring of $G$ by recursive calls on the graphs $G_h$ for all faces $h$ of $H$.

Finally, suppose that that $|\{v_1,v_3\}\cap V(B)|\le 1$.   By symmetry, we can assume that $v_1\not\in V(B)$.  Using Theorem~\ref{thm-mainalg}, we find a $3$-coloring $\psi'$ of $B+v_1v_2v_3v_4$ that
extends to a $3$-coloring of $G$ and whose restriction to $B$ is equal to $\psi$.  Note that either $\psi'(v_1)=\psi'(v_3)$ or $\psi'(v_2)=\psi'(v_4)$.  By symmetry, we can assume the former.
Suppose that $G$ contains a path $P$ of length at most $3$ joining $v_1$ with $v_3$ and disjoint from $\{v_2,v_4\}$, and let $K$ be the cycle $P+v_1v_2v_3$.  Since $|K|\le 5$ and $K$ is distinct
from the boundary cycles, it follows that $K$ is contractible, and thus $K$ bounds a face.  However, then $v_2$ has degree two, which is a contradiction.
Therefore, there exists no such path, and thus the graph $G'$ obtained from $G$ by identifying $v_1$ with $v_3$ is triangle-free.  Furthermore, since $\psi'$ extends to a $3$-coloring of $G$,
it follows that $\psi$ extends to a $3$-coloring of $G'$.  We can find such a $3$-coloring of $G'$ by a recursive call, and extend it to $G$ by giving $v_1$ and $v_3$ the color of the corresponding
vertex of $G'$.

To analyze the time complexity of the algorithm, note that in each case, we either call the quadratic algorithm of Corollary~\ref{cor-girth5}, or we
spend a linear time processing $G$ and recurse on several graphs $G_1$, \ldots, $G_k$,
such that the sum of the numbers of faces of $G_1$, \ldots, $G_k$ is at most the number of the faces of $G$.  Since the number of faces of $G$ is linear in the number of its
vertices, this implies that the time complexity of the algorithm is quadratic.
\end{proof}

\end{document}